\documentclass{article}
\usepackage{spconf,amsmath,graphicx,hyperref}
\usepackage{amsmath,amsfonts}
\usepackage{algorithmic}
\usepackage{algorithm}
\usepackage{array}
\usepackage[caption=false,font=normalsize,labelfont=sf,textfont=sf]{subfig}
\usepackage{amsmath}  
\usepackage{amssymb}  
\usepackage{amsthm}
\usepackage{bm}      
\usepackage{booktabs}
\usepackage{textcomp}
\usepackage{stfloats}
\usepackage{url}
\usepackage{verbatim}
\usepackage{graphicx}
\usepackage{cite}
\usepackage{xcolor}

\newtheorem{theorem}{Theorem}
\newtheorem{lemma}[theorem]{Lemma}      
\newtheorem{proposition}[theorem]{Proposition}

\ninept

\title{Low-Rank-Based Approximate Computation with Memristors}

\name{
    Binyu Lu$^{\star \dagger}$
    \qquad
    Matthias Frey$^{\star}$
    \qquad
    Stark Draper$^{\dagger}$
    \qquad
    Jingge Zhu$^{\star}$
}
  
\address{
    $^{\star}$ Department of Electrical and Electronic Engineering, University of Melbourne \\
    $^{\dagger}$ Department of Electrical and Computer Engineering, University of Toronto
}

\begin{document}

\maketitle

\begin{abstract}
    Memristor crossbars enable vector–matrix multiplication (VMM), and are promising for 
    low-power applications.
    However, it can be difficult to write the memristor conductance values exactly.
    To improve the accuracy of VMM, we propose a scheme based on low-rank matrix approximation.
    Specifically, singular value decomposition (SVD) is first applied to obtain a low-rank approximation of the target matrix, which is then factored into a pair of smaller matrices.
    Subsequently, a two-step serial VMM is executed, where the stochastic write errors are mitigated through step-wise averaging.
    To evaluate the performance of the proposed scheme, we derive a general expression for the resulting computation error and provide an asymptotic analysis under a prescribed singular-value profile, 
    which reveals how the error scales with matrix size and rank. 
    Both analytical and numerical results confirm the superiority of the proposed scheme compared with the benchmark scheme.
\end{abstract}

\section{Introduction}
\label{sec:intro}

With the growth of data-intensive workloads, particularly in fields such as machine learning and signal processing, the traditional von Neumann architecture faces significant challenge in latency and energy consumption.
This is due to frequent data movement between memory and processing units \cite{9563028, Boroumand_Google}.
In-memory computing (IMC) addresses this challenge by performing computations directly where the data resides.
It has emerged as a promising paradigm that offers outstanding energy efficiency and computation density \cite{silvano2025survey}.

IMC can physically be realized with memristor crossbars, which are two-dimensional (2D) arrays composed of row wordlines (WLs) and column bitlines (BLs).
A memristor cell located at each intersection.
A memristor is typically a non-volatile resistive device whose conductance can be tuned continuously by modulating the electric field and heat \cite{yao2020fully}.
There are three primary IMC uses of memristor crossbars: vector–matrix multiplication (VMM), Boolean logic computation, and Hamming distance computation \cite{10323169}.
Among these, VMM has attracted substantial interest due to its widespread use in deep neural networks (DNN) and in optimization \cite{jain2019neural}.
We focus on VMM in this paper.

In practice, memristor crossbars suffer from noisy writes.
For example, the programmed conductance values of the memristors may differ from the target values \cite{liu2018memristor}.
This induces noise in computation.
To mitigate the impact of such conductance variations, prior work has introduced redundancy to the system by applying coding-theoretic approaches \cite{Roth2019} or adaptive matrix mapping algorithms \cite{Lixue2018}.
Despite these efforts, studies aimed at enhancing the intrinsic robustness of analog computation remain limited.
The existing studies primarily focus on the standard VMM model and are largely limited to performance analysis.
For example, the authors in \cite{Dupraz2020Noisy} approximated the computation error of the analog VMM by a Gaussian distribution, and \cite{Dupraz2021ITW} presented the statistical characteristics of layer-by-layer computation errors in DNNs.
However, systematic approaches to reducing computation error beyond the standard VMM remain underexplored.

Motivated by this limitation, we propose a scheme that substantially reduces the computation error relative to a baseline scheme. 
First, singular value decomposition (SVD) is applied to construct a low-rank approximation of the target matrix, which is then factored into a pair of smaller matrices.
Then, VMMs are repeatedly performed on these smaller matrices in two serial steps, with step-wise averaging to suppress error resulting from the write operation.
We derive a general expression for the overall computation error and provide an asymptotic analysis under a prescribed singular value profile.
For $n\times n$ matrices, the error order improves from $\mathcal{O} (n^{2})$ to $\max \{\mathcal{O} (n^{3/2}), \mathcal{O} (n^{2-\alpha})\}$ in the asymptotic regime $n \to \infty$.
The parameter $\alpha \in (0, 1]$ describes how the matrix rank scales with $n$.
Simulation results validate the general performance analysis and demonstrate the superiority of the proposed scheme compared with the baseline scheme.

\section{System Model}

\subsection{Noiseless VMM on memristor crossbars}

Analog VMM on memristor crossbars can be expressed as  
\begin{equation} \label{noiselesscomputation}
	\mathbf{c} = \mathbf{b} \mathbf{A},
\end{equation}
where $\mathbf{b} \! \in \hspace{-1pt} \mathbb{R}^{1 \times m}$ is the input vector, $\mathbf{A} \! \in \hspace{-1pt} \mathbb{R}^{m \times n}$ is the target matrix, and $\mathbf{c} \! \in \hspace{-1pt} \mathbb{R}^{1 \times n}$ is the output vector.
Note that the bold lowercase letters denote row vectors in this paper.

Each entry $b_j$ of the vector $\mathbf{b}$ is transformed into a voltage proportional to its value,
and each entry $a_{j,k}$ of the matrix $\mathbf{A}$ is written to a memristor as the conductance value $g_{j,k}$ at the junction of WL $j$ and BL $k$, normalized by the feedback resistance $r_{\mathrm{T}}$ of a transimpedance amplifier (TIA) located at the end of each BL, i.e., $a_{j, k} = g_{j,k}/r_{\mathrm{T}}$.
The current $i_k$ at each BL is $i_k = \sum_{j=1}^m g_{j,k} b_j$.
The output TIA transforms the current $i_k$ into voltage $c_k = r_{\mathrm{T}} i_k$.
The product $\mathbf{b} \mathbf{A}$ is read as the voltage vector $\mathbf{c}$, at the grounded column BL outputs.

In this paper, we impose a total magnitude constraint on the target programming conductance values. It is formulated as
\begin{equation} \label{powerconstraint}
    \sum_{j=1}^{m} \sum_{k=1}^{n} g_{j,k}^2 \leq mn \rho,
\end{equation}
where $\rho $ is a constant determined by the intrinsic magnitude property of the memristors.
The constraint takes a form similar to the optimization objective used in \cite{Dupraz2021ITW}.

Since memristor conductances are non-negative, the matrix is typically expressed using two crossbars, one for positive coefficients and one for negative coefficients \cite{liu2018memristor}.
The final computational result is obtained through current subtraction in the analog domain.
We analyze the computation error for the positive crossbar.
The negative one can be analyzed analogously.
Combining them yields the total error expression.

\subsection{Baseline noisy computation scheme}
\label{2_2BaselineScheme}

Due to the fact that memristor programming is inherently imprecise, the actually programmed value, denoted by $a_{j,k}^{\prime}$, may deviate from the target value $a_{j,k}$.
Such errors lead to computation errors.
We write $a_{j, k}^{\prime} = (g_{j,k} + \Delta g_{j,k} )/r_{\mathrm{T}}$, where
$\Delta g_{j,k}$ is the write-time conductance variation.
This induces the coefficient-level perturbation  $E_{j,k} = \Delta g_{j,k}/ r_{\mathrm{T}}$ \cite{Chen2018noisymemristor, liu2018memristor}.
We model the programmed values as random variables as $A_{j,k}^{\prime} =  a_{j,k} + E_{j,k}$.
Throughout this paper, we assume that all memristors have identical fabrication quality,
meaning that $E_{j,k}$ are independent and identically distributed (i.i.d.) random variables with zero mean and variance $\sigma_e^2$.
In matrix form, $\mathbf{A}^{\prime} = \mathbf{A} + \mathbf{E}$,
where $\mathsf{Cov}(\mathsf{vec}(\mathbf{E})) =\sigma_e^2 \mathbf{I}_{mn}$, and the operator $\mathsf{vec}(\cdot)$ represents vectorization by row stacking.
Then, the computation actually performed by the hardware is expressed as
\begin{equation}\label{noisycomputation}
	\mathbf{c}^{\prime} = \mathbf{b} \left( \mathbf{A} + \mathbf{E} \right).
\end{equation}

In applications of interest, the matrix $\mathbf{A}$ changes much less frequently than does input $\mathbf{b}$, and can remain fixed entirely \cite{Roth2019}.
To evaluate the system performance, we compute the expected error over different input vectors $\mathbf{b}$, and over different realizations of memristor noise $\mathbf{E}$.
The later captures device-to-device variability.
Accordingly, the expected computation error for a fixed matrix $\mathbf{A}$ is
\begin{equation} \label{baselineexpression}
    \mathcal{E}^{\prime}
    =
    \mathbb{E}_{\mathbf{b}, \mathbf{E}} \!
    \left[
        \left\| \mathbf{c}' - \mathbf{c} \right\|_{\mathsf{F}}^2
    \right]
    =
    \mathbb{E}_{\mathbf{b}, \mathbf{E}} \! \left[ \left\| \mathbf{b} \mathbf{E} \right\|_{\mathsf{F}}^2 \right].
\end{equation}
Note that \eqref{baselineexpression} increases linearly with respect to the array size $mn$.
This can be substantial for large matrices.
To reduce the computation error, we now introduce an approximate computation scheme, based on the SVD and a serial two-step computation.

\section{Proposed Approximate computation scheme}

Our scheme replaces the baseline one-time large-scale VMM with multiple low-rank VMMs.
This can significantly reduce the expected computation error.
Specifically, we first decompose the large matrix into a pair of smaller matrices, perform parallel computations on each, and then average the results to improve accuracy.

\subsection{Rank-$k$ approximation of matrix $\mathbf{A}$}

Let the SVD of the matrix $\mathbf{A}$ be $\mathbf{A} = \mathbf{U} \mathbf{\Sigma} \mathbf{V}^{\mathsf{T}}$, in which $\mathbf{U}$ and $\mathbf{V}$ are orthogonal, and the rectangular matrix $\mathbf{\Sigma}$ contains the $r$ singular values, $\sigma_1 \geq \cdots \geq \sigma_r \geq 0$, along its diagonal, with $r \leq \min \{m, n\}$.
We partition the matrices $\mathbf{U}$, $\mathbf{V}$, and $\mathbf{\Sigma}$ as 
\begin{align}
    \mathbf{U} = \left[\mathbf{U}_k \ \mathbf{U}_\perp \right], \
    \mathbf{V} = \left[\mathbf{V}_k \ \mathbf{V}_\perp \right], \
    \mathbf{\Sigma} = \mathsf{diag} \left[ \mathbf{\Sigma}_k, \mathbf{\Sigma}_{r-k} \right],
\end{align}
where $\mathbf{U}_k$ and $\mathbf{V}_k$ consist of the first $k$ columns of $\mathbf{U}$ and $\mathbf{V}$, corresponding to the largest $k$ singular values in $\mathbf{\Sigma}_k$.
The best rank-$k$ approximation in terms of Frobenius norm is \cite{Eckart_Young_1936}
\begin{equation}
    \mathbf{A}_k = \mathbf{U}_k \mathbf{\Sigma}_k \mathbf{V}_k^{\mathsf{T}}.
\end{equation}
For in-memory implementation, we factor $\mathbf{A}_k$ as $\mathbf{A}_k = \mathbf{L} \mathbf{R}$,
where $\mathbf{L} = \mathbf{U}_k \mathbf{\Sigma}_k^{\frac{1}{2}}$ and $\mathbf{R} = \mathbf{\Sigma}_k^{\frac{1}{2}} \mathbf{V}_k^{\mathsf{T}}$.
This reduces the number of memristors used from $mn$ to $mk+nk$.
The low-rank truncation error introduced by the approximation is controlled by $k$, with $k \leq r$.

\subsection{Serial two-step computation}

Based on this best rank-$k$ approximation of $\mathbf{A}$, we now execute a two-step computation to approximate the result $\mathbf{c}_k = \mathbf{b A}_k = \mathbf{b L R}$, as shown in Fig.~\ref{dot-product-model}.

\begin{figure}[!t]
	\centering
	\includegraphics[width = 0.9 \linewidth]{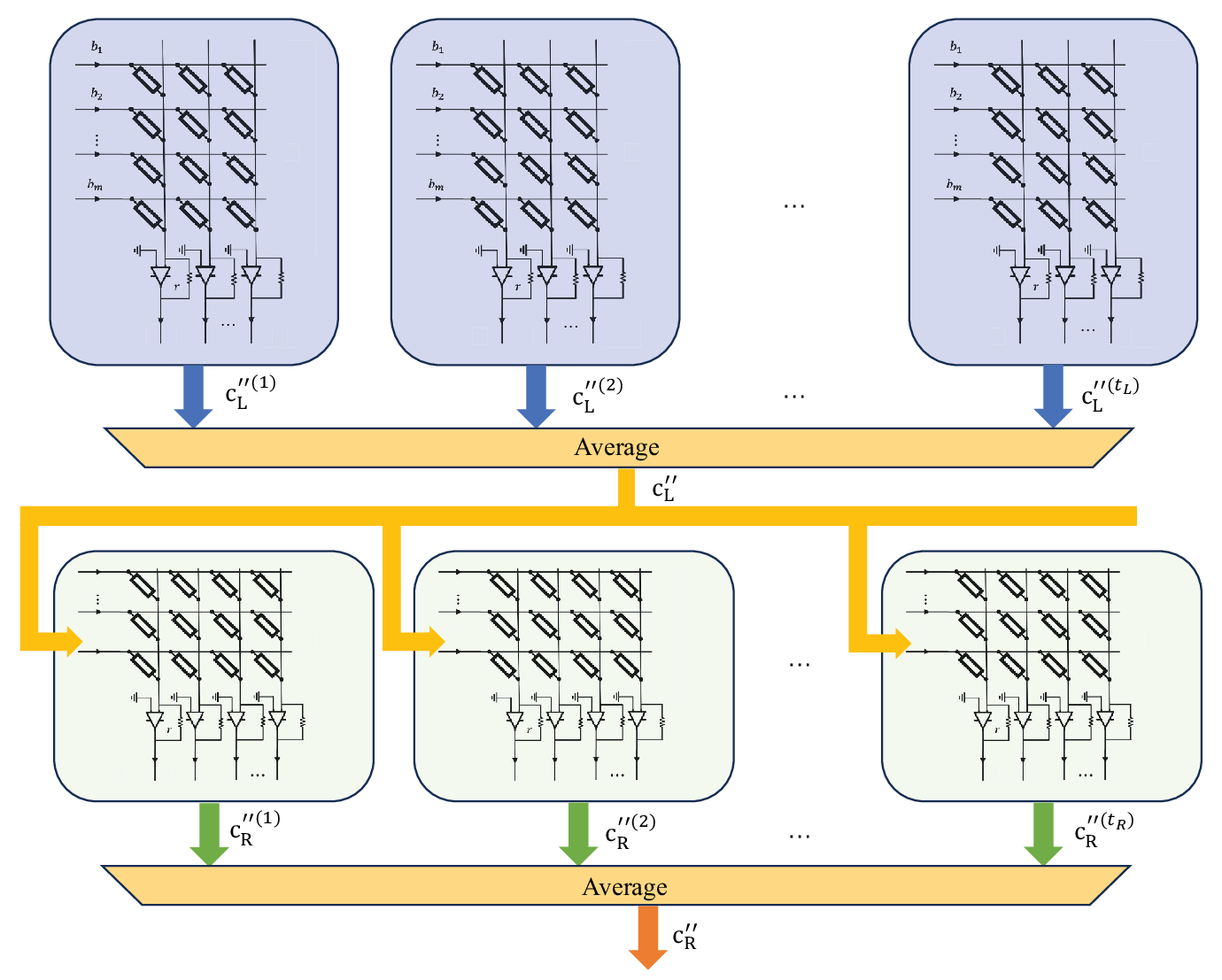}
	\vspace{-15pt}
	\caption{Serial two-step computation with memristor crossbars.\hspace{25pt}}
    \vspace{-10pt}
	\label{dot-product-model}
\end{figure}

Given a fixed budget of $mn$ memristors,
which matches the resource used by the baseline scheme described in Section~\ref{2_2BaselineScheme}.
The serial two-step computation uses $mk$ memristors for each computation in the first step and $nk$ for each computation in the second step.
As long as $mn > mk+nk$, we can repeat the computation in each of the two steps to reduce the computation error.
Let $t_\mathrm{L}, t_\mathrm{R} \in \mathbb{Z}_{+}$ denote the numbers of repetitions in the first and the second step, respectively.
They satisfy the constraint $t_{\mathrm{R}} \leq \frac{mn - mkt_{\mathrm{L}}}{nk}$, which is equivalent to $t_{\mathrm{L}} mk+ t_{\mathrm{R}} nk \leq mn$.

First, we repeat the programming of the matrix $\mathbf{L}$ across $t_{\mathrm{L}}$ distinct $m \times k$ memristor crossbar arrays.
Each repetition is affected by an independent noise realization $\mathbf{E}_{\mathrm{L}}^{(i)}$, $i \in [t_{\mathrm{L}}]$. 
The input vector $\mathbf{b}$ is converted to a voltage input.
Averaging the intermediate results of the $t_{\mathrm{L}}$ arrays, we obtain
\begin{equation}\label{stage-1}
    \mathbf{c}_{\mathrm{L}}^{\prime\prime} 
    =
    \frac{1}{t_{\mathrm{L}}} \sum_{i=1}^{t_{\mathrm{L}}} \mathbf{b} \left( \mathbf{L} + \mathbf{E}_{\mathrm{L}}^{(i)} \right).
\end{equation}

The averaged computation result from the first step is used as the input to the second step.
By repetition, the matrix $\mathbf{R}$ is programmed into $t_{\mathrm{R}}$ distinct $k \times n$ arrays.
Each matrix agian experience independent noise matrix $\mathbf{E}_{\mathrm{R}}^{(j)}$, $j \in [t_{\mathrm{R}}]$. 
The average output $\mathbf{c}^{\prime\prime}_{\mathrm{R}}$ of the second step is also the final result $\mathbf{c}^{\prime\prime}$, i.e.,
\begin{equation}\label{stage-2}
    \mathbf{c}^{\prime\prime} 
    =
    \mathbf{c}^{\prime\prime}_{\mathrm{R}} 
    = 
    \frac{1}{t_{\mathrm{R}}}
    \sum_{j=1}^{t_{\mathrm{R}}} 
    \mathbf{c}^{\prime\prime}_{\mathrm{L}} 
    \left( \mathbf{R} + \mathbf{E}_{\mathrm{R}}^{(j)} \right).
\end{equation}

\section{General error analysis}

In this section, we derive the forms of the computation error for the baseline computation scheme and the proposed approximate computation scheme, respectively.

\subsection{Error analysis for the baseline scheme}
\label{Section4_1}
Suppose the input $\mathbf{b}$ is a zero-mean random vector with covariance $\mathsf{Cov} (\mathbf{b}) = \sigma_b^2 \mathbf{I}_m$, 
and $\mathbb{E}_{\mathbf{b}} [\|\mathbf{b}\|_{\mathsf{F}}^{2}] < \infty$.
Also, suppose the noise $\mathbf{E}$ is a zero-mean random matrix, 
whose entries are i.i.d. with variance $\sigma_e^2$.
Moreover, $\mathbf{b}$ and $\mathbf{E}$ are independent with each other.

\begin{lemma} \label{lemma_2}
    Let $\mathbf{b} \in \mathbb{R}^{1 \times m}$ be a random vector with zero mean and covariance $\mathsf{Cov} (\mathbf{b}) = \sigma_b^2 \mathbf{I}_m$, 
    and let $\mathbf{M} \in \mathbb{R}^{m\times k}$ be any fixed matrix.
    Then, we have
    \begin{equation}\label{lemma_eq}
        \mathbb{E}_{\mathbf{b}} \! \left[ \left\| \mathbf{b} \mathbf{M} \right\|_\mathsf{F}^{2} \right]
        = 
        \mathsf{Tr} \! \left(\mathbf{M}\mathbf{M}^{\mathsf{T}} \, \mathbb{E} \! \left[ \mathbf{b}^{\mathsf{T}}\mathbf{b} \right] \right)
        = 
        \sigma_b^{2} \, \mathsf{Tr} \! \left( \mathbf{M}\mathbf{M}^{\mathsf{T}} \right).
    \end{equation}
\end{lemma}
\begin{proof}[Proof]
Use the trace identity $\|\mathbf{b}\mathbf{M}\|_{\mathsf{F}}^{2} = \mathsf{Tr} \! \left(\mathbf{M}^{\mathsf{T}}\mathbf{b}^{\mathsf{T}}\mathbf{b}\mathbf{M}\right)$ $= \mathsf{Tr} \! \left(\mathbf{M}\mathbf{M}^{\mathsf{T}}\mathbf{b}^{\mathsf{T}}\mathbf{b}\right)$, and the linearity of expectation and trace.
Then take the expectation, noting that $\mathbb{E}[\mathbf{b}^{\mathsf{T}}\mathbf{b}]=\mathsf{Cov}(\mathbf{b})=\sigma_b^{2}\mathbf{I}_m$, to obtain the second equality in \eqref{lemma_eq}.
\end{proof}

Based on Lemma~\ref{lemma_2}, we can now present the following theorem for the expected computation error of the baseline scheme.

\begin{theorem}\label{theorem2}
The expected computation error of the baseline scheme over input $\mathbf{b}$ and memristor noise $\mathbf{E}$ is given by
\begin{equation} \label{traditionalscheme}
    \mathcal{E}^{\prime} =
    \mathbb{E}_{\mathbf{b}, \mathbf{E}} \!
    \left[
        \left\| \mathbf{c}' - \mathbf{c} \right\|_{\mathsf{F}}^2
    \right]
    =
    \mathbb{E}_{\mathbf{b}, \mathbf{E}} \! \left[ \left\| \mathbf{b} \mathbf{E} \right\|_{\mathsf{F}}^2 \right] 
    = m \hspace{0.5pt} n \hspace{0.5pt} \sigma_e^2 \sigma_b^2,
\end{equation}
the order of which is $\mathcal{E}^{\prime} = \mathcal{O} (n^{2})$, if $m = \mathcal{O} (n)$ as $n \to \infty$.
\end{theorem}

\begin{proof}
Define
\begin{equation}
    f(\mathbf{b},\mathbf{E})
    \triangleq
    \|\mathbf{b}\mathbf{E}\|_{\mathsf{F}}^2
    = \mathsf{Tr} \! \left( \mathbf{E}\mathbf{E}^{\mathsf{T}} \mathbf{b}^{\mathsf{T}}\mathbf{b} \right)
     \geq 0 .
\end{equation}
By Tonelli's theorem for nonnegative measurable functions (Theorem 1.7.15 in \cite{tao2011introduction}), we may swap the order of integration since $f(\mathbf{b},\mathbf{E}) \geq 0$. The overall expectation is expressed as
\begin{align}\label{tonelli}
    \mathbb{E}_{\mathbf{b},\mathbf{E}}[f(\mathbf{b},\mathbf{E})]
    & = 
    \int_{\Omega_{\mathbf{b}}} \int_{\Omega_{\mathbf{E}}} 
    f(\mathbf{b},\mathbf{E}) 
    \mathrm{d} \mathbb{P}_{\mathbf{E}} \mathrm{d} \mathbb{P}_{\mathbf{b}}  \notag  \\
    & = 
    \int_{\Omega_{\mathbf{E}}} \int_{\Omega_{\mathbf{b}}} 
    f(\mathbf{b},\mathbf{E})
    \mathrm{d} \mathbb{P}_{\mathbf{b}} \mathrm{d} \mathbb{P}_{\mathbf{E}}.
\end{align}
Using the property of Frobenius norm, we note that
\begin{equation}
    f(\mathbf{b},\mathbf{E})=\|\mathbf{b}\mathbf{E}^{\mathsf{T}}\|_{\mathsf{F}}^2
    \leq 
    \|\mathbf{E}\|_{\mathsf{F}}^2 \|\mathbf{b}\|_{\mathsf{F}}^2 .
\end{equation}
Taking expectations and using independence, we have
\begin{equation}
    \mathbb{E}_{\mathbf{b}, \mathbf{E}} [f(\mathbf{b},\mathbf{E})]
    \leq 
    \mathbb{E}_{\mathbf{E}} \! \left[ \|\mathbf{E}\|_{\mathsf{F}}^2 \right]
    \mathbb{E}_{\mathbf{b}} \! \left[ \|\mathbf{b}\|_{\mathsf{F}}^2 \right]
    = m n \sigma_e^2\cdot m\sigma_b^2
    < \infty.
\end{equation}
Thus $f$ is integrable on the product space. By Fubini's theorem for integrable functions ((Theorem 1.7.21 in \cite{tao2011introduction}),
the iterated expectations in \eqref{tonelli} are equal to each other and to the joint expectation, and are finite.
Now, for any fixed $\mathbf{E}$, by Lemma ~\ref{lemma_2}, we obtain
\begin{align}
    \mathbb{E}_{\mathbf{b}} \! \left[ f(\mathbf{b},\mathbf{E}) \right]
    = 
    \sigma_b^{2} \mathsf{Tr} \! \left( \mathbf{E}\mathbf{E}^{\mathsf{T}} \right),
\end{align}
where $\mathbb{E}_{\mathbf{b}} [\mathbf{b}\mathbf{b}^{\mathsf{T}}] = \sigma_b^2 \mathbf{I}_m$.
Taking the outer expectation over $\mathbf{E}$,
\begin{align}
    \mathcal{E}^{\prime}
    = 
    \mathbb{E}_{\mathbf{E}} \! \left[\sigma_b^{2} \mathsf{Tr} \left(\mathbf{E}\mathbf{E}^{\mathsf{T}}\right)\right]
    = 
    \sigma_b^{2}\,\mathbb{E}_{\mathbf{E}} \! \left[\|\mathbf{E}\|_{\mathsf{F}}^{2}\right]
    = 
    m n \sigma_e^{2} \sigma_b^{2}.
\end{align}
If $m=\mathcal O(n)$, $n \to \infty$, and $\sigma_b^2,\sigma_e^2$ do not scale with $n$, then
\begin{equation}
    \mathcal E'=mn\sigma_b^2\sigma_e^2=\mathcal O(n^2).
\end{equation}
\end{proof}

\subsection{Error analysis for the approximate computation}

In addition to the assumptions for Section \ref{Section4_1}, 
we assume that $\mathbf{E}_{\mathrm{L}}^{(i)} \! \in \mathbb{R}^{m \times k}$ and $\mathbf{E}_{\mathrm{R}}^{(j)} \! \in \mathbb{R}^{k \times n}$ are zero-mean random matrices with i.i.d. entries and variances $\sigma_{\mathrm{L}}^2$ and $\sigma_{\mathrm{R}}^2$, respectively, 
for all $i \in [t_{\mathrm{L}}]$  and $j \in [t_{\mathrm{R}}]$.
Note that $\mathbf{b}$, $\mathbf{E}_{\mathrm{L}}^{i}$, and $\mathbf{E}_{\mathrm{R}}^{j}$ are independent of one another.

Combining \eqref{noiselesscomputation} and \eqref{stage-2}, the expected computation error of the proposed scheme is given in \eqref{ErrorAC}, shown at the top of next page.
\begin{figure*}[!t]
\begin{align} \label{ErrorAC}
    \mathcal{E}^{\prime\prime} 
    & = 
    \mathbb{E}_{\mathbf{b}, \mathbf{E}_{\mathrm{L}}, \mathbf{E}_{\mathrm{R}}} 
    \left[ 
        \left\| \mathbf{c}^{\prime\prime} - \mathbf{c} \right\|_{\mathsf{F}}^2 
    \right] \notag \\
    & =  
    \mathbb{E}_{\mathbf{b}, \mathbf{E}_{\mathrm{L}}, \mathbf{E}_{\mathrm{R}}} 
    \left[
        \left\|
            \mathbf{b} \left( \mathbf{A}_k - \mathbf{A} \right)
            +
            \mathbf{b} \left( \frac{1}{t_{\mathrm{L}}} \sum_{i=1}^{t_{\mathrm{L}}} \mathbf{E}_{\mathrm{L}}^{(i)} \right) \mathbf{R} 
            +
            \mathbf{b} \mathbf{L} \left( \frac{1}{t_{\mathrm{R}}} \sum_{j=1}^{t_{\mathrm{R}}} \mathbf{E}_{\mathrm{R}}^{(j)} \right)
            +
            \frac{1}{t_{\mathrm{L}} t_{\mathrm{R}}} \mathbf{b} \left( \sum_{i=1}^{t_{\mathrm{L}}} \mathbf{E}_{\mathrm{L}}^{(i)} \right) \left( \sum_{j=1}^{t_{\mathrm{R}}} \mathbf{E}_{\mathrm{R}}^{(j)} \right)
        \right\|_{\mathsf{F}}^2
    \right].
\end{align}
\vspace{-5pt}
\hrule
\vspace{-10pt}
\end{figure*}
We present the following theorem to analyze the approximate scheme computation error.

\begin{theorem}
    The expected computation error of the proposed approximate computation scheme is 
    \begin{align} \label{approximate_error}
        \mathcal{E}^{\prime\prime}
        & = \\
        & \;
        \sigma_b^2 
        \left(
            \sum_{i=k+1}^r \sigma_i^2
            + 
            \left( \frac{m \sigma_{\mathrm{L}}^2}{t_{\mathrm{L}}} + \frac{n \sigma_{\mathrm{R}}^2}{t_{\mathrm{R}}}\right) \mathsf{Tr} \! \left(\mathbf{\Sigma}_k\right) + \frac{m k n \sigma_{\mathrm{L}}^2 \sigma_{\mathrm{R}}^2}{t_{\mathrm{L}} t_{\mathrm{R}}}
        \right). \notag
    \end{align}
\end{theorem}

\begin{proof}[Proof]
    Since $\mathbf{E}_{\mathrm{L}}^{(i)}$ and $\mathbf{E}_{\mathrm{R}}^{(j)}$ are  random matrices whose entries are i.i.d.with zero mean and variances $\sigma_{\mathrm{L}}^2$ and $\sigma_{\mathrm{R}}^2$, respectively,
    the averaged noise matrices are
    \begin{equation} \label{eq:average_e}
        \bar{\mathbf{E}}_{\mathrm{L}} 
        \triangleq
        \frac{1}{t_{\mathrm{L}}} \sum_{i=1}^{t_{\mathrm{L}}} \mathbf{E}_{\mathrm{L}}^{(i)}, 
        \; \;
        \bar{\mathbf{E}}_{\mathrm{R}} 
        \triangleq 
        \frac{1}{t_{\mathrm{R}}} \sum_{j=1}^{t_{\mathrm{R}}} \mathbf{E}_{\mathrm{R}}^{(j)},
    \end{equation}
    Their entries are i.i.d. with zero mean and reduced variances ${\sigma_{\mathrm{L}}^2}/{t_{\mathrm{L}}}$ and ${\sigma_{\mathrm{R}}^2}/{t_{\mathrm{R}}}$, respectively.

    Then, the computation error can be expressed as the sum of four terms as follows:
    \begin{equation}
        \mathbf{c}^{\prime\prime} - \mathbf{c}
        = 
        \underbrace{\mathbf{b} \left( \mathbf{A}_k - \mathbf{A} \right)}_{\mathbf{c}_1} 
        + 
        \underbrace{\mathbf{b} \bar{\mathbf{E}}_{\mathrm{L}} \mathbf{R}}_{\mathbf{c}_2} 
        + 
        \underbrace{\mathbf{b} \mathbf{L} \bar{\mathbf{E}}_{\mathrm{R}}}_{\mathbf{c}_3} 
        + 
        \underbrace{\mathbf{b} \left( \bar{\mathbf{E}}_{\mathrm{L}} \bar{\mathbf{E}}_{\mathrm{R}} \right)}_{\mathbf{c}_4}.
    \end{equation}
    If we expand $\| \mathbf{c}^{\prime\prime}- \mathbf{c} \|_{\mathsf{F}}^2$, then the cross-terms have the form $\left< \mathbf{c}_i, \mathbf{c}_j \right> = \mathbf{b} \mathbf{M}_i \mathbf{M}_j^{\mathsf{T}} \mathbf{b}^{\mathsf{T}}$, and
    $\mathbf{M}_1 \triangleq \mathbf{A}_k-\mathbf{A}$, $\mathbf{M}_2 \triangleq \bar{\mathbf{E}}_u \mathbf{R}$, $\mathbf{M}_3 \triangleq \mathbf{L} \bar{\mathbf{E}}_v$, and $\mathbf{M}_4 \triangleq \bar{\mathbf{E}}_u \bar{\mathbf{E}}_v$.
    The conditional expectation is $\mathbb{E}_{\mathbf{b}} \left[ \left< \mathbf{c}_i, \mathbf{c}_j \right> \mathop{|} \mathbf{M}_i, \mathbf{M}_j \right] = \sigma_b^2 \mathsf{Tr}(\mathbf{M}_i \mathbf{M}_j^{\mathsf{T}})$.
    Since $\mathbf{M}_i$ and $\mathbf{M}_j$ are independent, the expectation of cross terms are all zero values, i.e., $\mathbb{E}_{\mathbf{b}, \mathbf{M}_i, \mathbf{M}_j} [ \langle \mathbf{c}_i, \mathbf{c}_j \rangle ] = 0.$
    Hence, the error $\mathcal{E}^{\prime\prime}$ consists of four independent terms as
    \begin{equation} \label{eq:sum_terms}
        \mathcal{E}^{\prime\prime} 
        = 
        \mathbb{E}_{\mathbf{b}, \mathbf{E}_{\mathrm{L}}, \mathbf{E}_{\mathrm{R}}}
        \Big[
            \left\| \mathbf{c}_1 \right\|^2_{\mathsf{F}}
            + 
            \left\| \mathbf{c}_2 \right\|^2_{\mathsf{F}}
            + 
            \left\| \mathbf{c}_3 \right\|^2_{\mathsf{F}}
            + 
            \left\| \mathbf{c}_4 \right\|^2_{\mathsf{F}}
        \Big].
    \end{equation}
    In the rest of proof, the subscripts under $\mathbb{E}$ are omitted for simplicity.
    Based on Lemma~\ref{lemma_2}, the first three terms are
    \begin{align} 
        \mathbb{E} \left[ \left\| \mathbf{c}_1 \right\|^2_{\mathsf{F}} \right]
        & = 
        \sigma_b^2 
        \left\| \mathbf{A} - \mathbf{A}_k \right\|_{\mathsf{F}}^2 = \sigma_b^2 \sum_{i=k+1}^r \sigma_i^2. \label{eq:T1_final} \\ 
        \mathbb{E} \left[ \left\| \mathbf{c}_2 \right\|^2_{\mathsf{F}} \right]
        & = 
        \sigma_b^2 \, \mathsf{Tr} \! \left( \mathbf{R}^{\mathsf{T}} 
        \mathbb{E} \! \left[ \bar{\mathbf{E}}_{\mathrm{L}}^{\mathsf{T}} \bar{\mathbf{E}}_{\mathrm{L}} \right] \mathbf{R} \right) 
        = \sigma_b^2 \frac{m \sigma_{\mathrm{L}}^2}{t_{\mathrm{L}}} \| \mathbf{R} \|_\mathsf{F}^2 .\label{eq:T2_final}  \\ 
        \mathbb{E} \left[ \left\| \mathbf{c}_3 \right\|^2_{\mathsf{F}} \right]
        & = \sigma_b^2 \, \mathsf{Tr} \! \left( \mathbf{L}^{\mathsf{T}} \mathbb{E} \! \left[ \bar{\mathbf{E}}_{\mathrm{R}}^{\mathsf{T}} \bar{\mathbf{E}}_{\mathrm{R}} \right] \mathbf{L} \right) 
        = \sigma_b^2 \frac{n \sigma_{\mathrm{R}}^2}{t_{\mathrm{R}}} \| \mathbf{L} \|_\mathsf{F}^2. \label{eq:T3_final}
    \end{align}
    The first term is the low-rank truncation error.
    The fourth term is the accumulated noise.
    Similar to the proof of Theorem ~\ref{theorem2}, it can be expressed as
    \begin{equation} \label{eq:T4_final}
    \begin{aligned}
        \mathbb{E} \left[ \left\| \mathbf{c}_4 \right\|^2_{\mathsf{F}} \right]  
        & = 
        \sigma_b^2 \, \mathbb{E}   \left[ \left\| \bar{\mathbf{E}}_{\mathrm{L}} \bar{\mathbf{E}}_{\mathrm{R}} \right\|_\mathsf{F}^2 \right] \\
        & = \sigma_b^2 \, \mathbb{E} \! \left[ \mathsf{Tr} \left( \bar{\mathbf{E}}_{\mathrm{L}} \bar{\mathbf{E}}_{\mathrm{R}} \bar{\mathbf{E}}_{\mathrm{R}}^{\mathsf{T}} \bar{\mathbf{E}}_{\mathrm{L}}^{\mathsf{T}} \right) \right] \\
        & = \sigma_b^2 \, \mathsf{Tr} \left( \mathbb{E}_{\bar{\mathbf{E}}_{\mathrm{L}}} \! \left[ \bar{\mathbf{E}}_{\mathrm{L}}^{\mathsf{T}} \bar{\mathbf{E}}_{\mathrm{L}} \right] 
        \mathbb{E}_{\bar{\mathbf{E}}_{\mathrm{R}}} \! \left[ \bar{\mathbf{E}}_{\mathrm{R}} \bar{\mathbf{E}}_{\mathrm{R}}^{\mathsf{T}} \right] \right) \\
        & = \sigma_b^2  \mathsf{Tr} \left( \frac{m \sigma_{\mathrm{L}}^2}{t_{\mathrm{L}}} \mathbf{I}_k \frac{n \sigma_{\mathrm{R}}^2}{t_{\mathrm{R}}} \mathbf{I}_k \right) \\
        & = \sigma_b^2 \frac{n \sigma_{\mathrm{R}}^2}{t_{\mathrm{R}}} \frac{m k \sigma_{\mathrm{L}}^2}{t_{\mathrm{L}}}.
    \end{aligned}
    \end{equation}
    Using the fact that $\|\mathbf{L}\|_\mathsf{F}^2 = \|\mathbf{R}\|_\mathsf{F}^2 = \mathsf{Tr}(\mathbf{\Sigma}_k)$,
    and substituting  \eqref{eq:T1_final}, \eqref{eq:T2_final}, \eqref{eq:T3_final} and \eqref{eq:T4_final} into \eqref{eq:sum_terms}, 
    we obtain the final expression \eqref{approximate_error}.
\end{proof}

Although the rank-$k$ approximation introduces an extra truncation error, 
it shrinks the crossbar arrays involved from size $m\times n$ to $m\times k$ and $n\times k$.
The stochastic error, which grows with the number of memristors used, drops substantially.
The two-step small-size multiplication also enables repetitions and averaging, reducing variance by $1/t_{\mathrm{L}}$ and $1/t_{\mathrm{R}}$.
By choosing $k$, $t_{\mathrm{L}}$, and $t_{\mathrm{R}}$ to balance truncation and variance, we can improve the performance vis-a-vis the baseline.

\section{Asymptotic analysis for a particular singular-value profile}

\subsection{Asymptotic error analysis}

We consider a class of rank-$r$  matrices, denoted by $\mathcal{A}_h (\lambda)$.
The class consists of matrices whose singular values follow a harmonic decay profile with the parameter $\lambda$.
Formally, $\mathcal{A}_h (\lambda) = \left\{ \mathbf{A} \in \mathbb{R}^{m \times n} : \mathsf{rank}(\mathbf{A}) = r, \sigma_i(\mathbf{A}) = \tfrac{\lambda}{i}, i \in [r] \right\}$.

For any $\mathbf{A}\in\mathcal{A}_h(\lambda_0)$ with $\lambda_0 \le \frac{\sqrt{6mn\,\rho}}{\pi r_{\mathrm{T}}}$, the total squared magnitude of the programmed values is
\begin{align} \label{Apowersum}
    \sum_{i = 1}^{m} \sum_{j=1}^{n} g_{ij}^2 
    & = 
    r_{\mathrm{T}}^2 \, \mathsf{Tr} \! \left( \mathbf{U} \mathbf{\Sigma} \mathbf{V}^\mathsf{T} \mathbf{V}  \mathbf{\Sigma} \mathbf{U}^\mathsf{T} \right) \notag \\
    & = 
    r_{\mathrm{T}}^2 \sum_{i = 1}^{r} \frac{1}{i^2} \lambda_0^2
    \leq 
    r_{\mathrm{T}}^2 \sum_{i = 1}^{\infty} \frac{1}{i^2} \lambda_0^2 
    = 
    \frac{\pi^2}{6} \lambda_0^2 r_{\mathrm{T}}^2 
    \leq mn\rho,
\end{align}
where the first inequality is based on the result of Basel problem \cite{havil2010gamma}.
The second inequality of \eqref{Apowersum} ensures that $\mathbf{A}$ satisfies the total magnitude constraint in \eqref{powerconstraint}.

\begin{theorem} \label{Theorem_04}
Assume that $m = n$, $n \to \infty$, the approximation rank $k$ and the rank $r$ take the form of $k = c_1 \cdot r^{\beta}$, $(0<c_1 \leq 1, 0 < \beta \leq 1)$, and $r = c_2 \cdot n^{\alpha}$, $(0<c_2 \leq1, 0<\alpha\leq1)$. 
The asymptotic expression of approximate computation error is given by (\ref{asymptotic_error}).
\begin{figure*}[!t]
\begin{align} \label{asymptotic_error}
    \mathcal{E}^{\prime\prime}
    & \leq
    \sigma_b^2 \left[ 
        \lambda^2 
            \left(\frac{1}{c_1} \frac{1}{c_2^{\beta} n^{\alpha \beta}} - \frac{1}{c_2 n^{\alpha}}\right)
     + 4c_1  c_2^{\beta} \lambda \left(\sigma_{\mathrm{L}}^2 
     +
     \sigma_{\mathrm{R}}^2 \right) n^{\alpha \beta} \left( \alpha \beta \ln n + \frac{1}{2c_1  c_2^{\beta} n^{\alpha \beta}} 
     + \gamma_{\mathrm{EM}} \right)
     + 4 c_1^3  c_2^{3\beta} n^{3 \alpha \beta} \sigma_{\mathrm{L}}^2 \sigma_{\mathrm{R}}^2 \right] .
\end{align}
\hrule
\vspace{-10pt}
\end{figure*}
\end{theorem}

\begin{proof}  
    When $k \to \infty, r\to \infty$, the asymptotic expression of the low-rank truncation error is
    \begin{align} \label{Tail_Square_sum}  
         S_{\mathrm{t}}
         & \triangleq 
        \sum_{i=k+1}^r \sigma_i^2 
        = 
        \sum_{i=k+1}^{r} {\frac{\lambda^2}{i^2}} \leq \lambda^2 \sum_{i=k}^{r-1} \int_i^{i+1} \frac{dt}{t^2} 
        = 
        \lambda^2 \int_k^{r} \frac{dt}{t^2} \notag \\
        & = 
        \lambda^2 \left(  \frac{1}{k} - \frac{1}{r} \right) =  \lambda^2 \left(  \frac{1}{c_1r^{\beta}} - \frac{1}{r} \right)  .
    \end{align}
    Since $ \mathsf{Tr} \! \left( \mathbf{\Sigma}_k \right)$ takes the form of a partial sum of the harmonic series, its expansion is given by \cite{Boas01101971}
    \begin{equation} \label{TopSum}
        S_{\mathrm{h}}
        \triangleq
        \mathrm{Tr} \! \left(\boldsymbol{\Sigma}_k\right) 
        = 
        \sum_{i=1}^k \frac{\lambda}{i}
        \leq 
        \lambda \left(\ln k + \gamma_{\mathrm{EM}} + \frac{1}{2k}\right).
    \end{equation}
    where $\gamma_{\mathrm{EM}} \approx 0.5772$ is the Euler-Mascheroni constant. 
    Substituting $k = c_1 \cdot r^{\beta}$ and $r = c_2 \cdot n^{\alpha}$ into (\ref{Tail_Square_sum}) and (\ref{TopSum}), and plugging these two equations into (\ref{approximate_error}) yields the bound in (\ref{asymptotic_error}).
    In this step, we also invoke the assumption $m=n$ from the theorem statement to eliminate the dependence on $m$ and obtain the expression purely in terms of $n$.
\end{proof}

Based on above theorem, we give the optimal low-rank $k$ that minimizes the asymptotic error in the following proposition.

\begin{proposition} \label{Proposition_5}
    Assume that $m=n$, $n \to \infty$. 
    Consider any matrix $\mathbf{A} \in \mathcal{A}_h (\lambda_0)$ with  $\lambda_0 \leq  \frac{\sqrt{6mn \rho}}{ \, \pi r_{\mathrm{T}}}$ and rank $r = c_2 n^{\alpha}$, $(0<c_2 \leq1)$,
    for some  fixed $\alpha \in (0, 1]$.
    Then, the optimal low-rank parameter of the proposed scheme is
    \begin{equation}
        \beta^{\star} 
        =
        \min \left\{1, \tfrac{1}{2 \alpha} \right\},
    \end{equation}
    and the minimized scaling of the asymptotic computation error is
\begin{equation}\label{minerrorasym}
        \mathcal{E}^{\prime\prime}_{\min}
        =
        \begin{cases}
        \mathcal{O} \hspace{0.5pt} ( n^{2-\alpha} ),    & \alpha \in (0, \frac{1}{2}), \\
        \mathcal{O} \hspace{0.5pt} ( n^{\frac{3}{2}} ), & \alpha \in [\frac{1}{2}, 1].
        \end{cases}
    \end{equation}
\end{proposition}

\begin{proof}[Proof]
Based on Theorem~\ref{Theorem_04}, the asymptotic error of the proposed scheme is bounded by
    \begin{equation}\label{asymupperbound}
        \mathcal{E}^{\prime\prime} 
        \leq
        \mathcal{O} \! \left( n^{2 - \alpha \beta} \right)
        +
        \mathcal{O} \! \left( n^{\alpha \beta + 1} \ln n \right)
        +
        \mathcal{O} \! \left( n^{3 \alpha \beta} \right).
    \end{equation}
    To obtain the optimal scaling, we minimize the dominant term of \eqref{asymupperbound} as $n \to \infty$, which yields the minimized asymptotic error.
    
    Let $x \triangleq \alpha \beta \in (0, \alpha]$. 
    The bound in \eqref{asymupperbound} becomes $\mathcal{O}(n^{2-x}) + \mathcal{O}(n^{1+x}\ln n) + \mathcal{O}(n^{3x})$.
    Since $ \ln n = o(n^{\varepsilon}) $ for every $\varepsilon >0$,
    the dominant component is $\max\{n^{2-x},n^{1+x},n^{3x}\}$. 
    To obtain the minimum error bound, we minimize $\max \{2-x, 1+x, 3x\}$ over $x \in (0,\alpha]$.
    When  $\alpha \in (\tfrac{1}{2}, 1]$, the minimum value of the bound is achieved at the balance of the three components $2-x = 1+ x = 3x$, i.e., $x = \tfrac{1}{2}$.
    When $\alpha \in (0, \tfrac{1}{2}]$, the feasible region for $x$ is $(0, \tfrac{1}{2}]$ since $\beta \in (0,1]$, which yields $2-x \geq 1+ x \geq 3x$.
    Thus, the dominant term is always $ n^{2-x}$, and the minimum value is achieved at $x = \alpha$.
    Combining the above two cases, we obtain $x^\star=\min\{\alpha, \frac{1}{2}\}$. 
    
    Given $\alpha$, we can derive the optimal $\beta^{\star} $ from $x^\star=\min\{\alpha, \frac{1}{2}\}$.
    If $\alpha \in (\tfrac{1}{2}, 1]$, the optimal $x^{\star} = \tfrac{1}{2}$, which yields $\beta^{\star} = \tfrac{1}{2\alpha}$ and $\mathcal{E}^{\prime\prime}_{\min} = \mathcal{O} (n^{\frac{3}{2}})$.
    If $\alpha \in (0, \tfrac{1}{2}]$, the optimal $x^{\star} = \alpha$, which yields $\beta^\star=1$, and $\mathcal{E}^{\prime\prime}_{\min} = \mathcal{O} (n^{2-\alpha})$.
\end{proof}

Proposition~\ref{Proposition_5} indicates that the order of the minimum asymptotic error of the proposed scheme in \eqref{minerrorasym} is smaller than that of the baseline scheme, $\mathcal{E}^{\prime} = \mathcal{O} (n^{2})$, demonstrating that the proposed scheme outperforms the baseline scheme.

\subsection{A numerical example}

\begin{figure}[!t]
	\centering
    \vspace{-5pt}
	\includegraphics[width = 0.975 \linewidth]{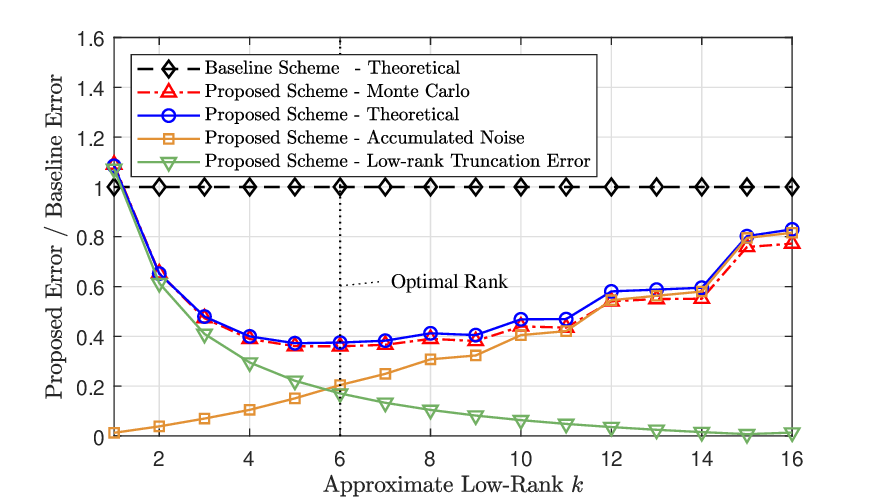}
	\vspace{-10pt}
	\caption{Comparison of proposed scheme and baseline scheme.\hspace{30pt}}
    \vspace{-10pt}
	\label{numericalresults}
\end{figure}

In Fig.~\ref{numericalresults}, we use a numerical example to compare the computation error of the proposed scheme with the baseline scheme. 
The matrix is $\mathbf{A}\in\mathbb{R}^{100\times 100}$ with singular-value profile $\sigma_i=\lambda/i$ and $r=16$. 
Noise variance is $\sigma_e^2=\sigma_L^2=\sigma_R^2=0.05$ and $\sigma_b^2=3$. 
The black curve shows the normalized baseline error, which is set to $1$. 
For each $k$, $10,000$ Monte Carlo trials are performed to validate the analytical expression in~\eqref{approximate_error}. 
This illustrates the effectiveness of proposed scheme compared with the baseline scheme.
For small $k$, the truncation error dominates, so increasing $k$ reduces the error. 
For larger $k$, the benefits of approximation diminish while the effective matrix size grows and the accumulated noise becomes dominant, which causes the error to rise. 
This tradeoff implies the existence of an optimal rank $k$ that minimizes the computation error under the fixed memristor budget.


\bibliographystyle{IEEEbib}
\bibliography{refs}

\end{document}